\documentclass[a4paper, 11pt]{article}

\usepackage{amsmath,amssymb,amsthm}

\usepackage[a4paper, margin=1.2in]{geometry}

\newtheorem{theorem}{Theorem}%[section]
\newtheorem{proposition}[theorem]{Proposition}
\newtheorem{lemma}[theorem]{Lemma}

\newtheorem{definition}[theorem]{Definition}

\usepackage{makeidx}
\usepackage{textcomp}
\usepackage{graphicx,graphics}
\usepackage{psfrag}
\usepackage{epsfig}

\newcommand{\vbl}{\ensuremath{{\rm vbl}}}

\newcommand{\poly}{\ensuremath{{\textup{\rm poly}}}}
\newcommand{\cost}{\ensuremath{{\textup{\rm cost}}}}
\newcommand{\Neg}{\ensuremath{{\textup{\rm Neg}}}}

\newcommand{\ignore}[1]{}

\usepackage{cancel}
\usepackage{algorithm}
\usepackage{algorithmic}

\newcommand{\UniquekSAT}[1]{\ensuremath{\textsc{Unique-#1-SAT}}}
\newcommand{\BkS}[1]{\ensuremath{\textsc{Ball-}#1\textsc{-SAT}}}

\newcommand{\BdkCSP}[2]{\ensuremath{\textsc{Ball-}(#1,#2)\textsc{-CSP}}}
\newcommand{\DBS}{\ensuremath{\textsc{Double-Ball-SAT}}}
\newcommand{\dbs}{\texttt{\textup{double-ball-search}}}

\title{Using Constraint Satisfaction To Improve Deterministic $3$-SAT} 

\author{
  Konstantin Kutzkov \\
  IT University of Copenhagen, Denmark\\
  \texttt{kutzkov@googlemail.com}
  \and
Dominik Scheder\\
  ETH Z\"urich\\
  \texttt{dscheder@inf.ethz.ch} 
}
\begin{document}

\maketitle

\begin{abstract}
  We show how one can use certain deterministic algorithms for
  higher-value constraint satisfaction problems (CSPs) to
  speed up deterministic local search for $3$-SAT. This
  way, we improve the deterministic worst-case running
  time for $3$-SAT to $O(1.439^n)$. 
\end{abstract}

\section{Introduction}

Among NP-complete problems, boolean satisfiability, short SAT, is
perhaps the most intensively studied, rivaled only by graph
colorability. Within SAT, the case of $k$-SAT, in which the input is
restricted to formulas containing only clauses of size at most $k$, has
drawn most attention, in particular the case $k=3$. The currently best
algorithms for $3$-SAT are based on two rather different ideas.  The
first is local search. In~\cite{Schoening99}, Sch\"oning gives an
extremely simple randomized local search algorithm, and its running
time of $O^*( (4/3)^n)$ is still close to the best known running times
(here, $n$ is the number of variables, and we use the notation $O^*$
to suppress factors that are polynomial in $n$).  The
second idea is to process the variables of $F$ in random order,
assigning them $0$ or $1$ randomly, unless one choice is ``obviously
wrong''. This was introduced by Paturi, Pudl{\'a}k, and
Zane~\cite{ppz}, achieving a running time of $O^*(1.59^n)$. By using a
less obvious notion of ``obviously wrong'' (and a much more
complicated analysis), Paturi, Pudl{\'a}k, Saks, and Zane~\cite{ppsz}
significantly improved this, however not beating Sch\"oning's bound of
$(4/3)^n$. Iwama and Tamaki combined these two approaches to obtain a
running time of $O^*(1.3238^n)$.  Rolf~\cite{rolf05} improved the
analysis of that algorithm and showed that its running time is
$O^*(1.32216^n)$, the currently best bound.\\

Deterministic algorithms for $3$-SAT do not achieve these running
times. The currently best deterministic algorithms can all be seen as
attempts to derandomize Sch\"oning's local search algorithm.  The
first attempt is by Dantsin et al.~\cite{dantsin} and is based on a
simple recursive local search algorithm, combined with a construction
of covering codes. This achieves a running time of
$O^*(1.5^n)$.Dantsin et al. also show how the recursive algorithm can
be improved to achieve an overall running time of $O^*(1.481^n)$.
Subsequent papers (Brueggeman and Kern~\cite{BK04},
Scheder~\cite{Scheder08}) improve the local search algorithm by giving
more sophisticated branchings.  This paper also improves the running
time by improving the recursive local search algorithm, but is still
qualitatively different from previous work: We show that, under
certain circumstances, one can translate those $3$-CNF formulas that
constitute the worst-case for previous algorithms into a constraint
satisfaction problem (CSP) with more than two values (boolean
satisfiability problems are CSPs with two values), which can be solved
quite efficiently.
\ignore{
The hunt for the lowest base is not an idle exercise.  We want to give
three reasons for this belief. First, there is currently no
``obvious'' candidate for $s_k$. Think of Sch\"oning's
algorithm~\cite{Schoening99}. Its running time is $O^*((4/3)^n)$, and
$4/3$ is a beautiful number. If the ten years that have passed since
its publication had not seen any improvement over $4/3$, one would
guess some researchers would have conjectured $s_3$ to be $4/3$.
However, Hofmeister, Sch\"oning, Schuler and Watanabe~\cite{HSSW02}
gave the first running time below $(4/3)^n$, namely $O^*(1.3302^n)$.
This improvement is rather small, but still very valuable because it
shows that $4/3$ is not the last word. Daniel Rolf~\cite{rolf05} gives
the currently best randomized running time for $3$-SAT, namely
$O^*(1.32217^n)$. This is clearly not a beautiful number. Therefore,
further work is necessary if we want to have a good candidate for $s_3$.\\

Second, currently randomized algorithms for $3$-SAT are ahead of
deterministic algorithms by a wide margin. Let us denote by $d_3$ the
deterministic analog of $s_3$. We know that $s_3 \leq 1.32217$ and
$d_3 \leq 1.439$. Can one close that gap? Daniel
Rolf~\cite{rolf-unique} has shown that if the input is restricted to
$3$-CNF formulas with at most one satisfying assignment (sometimes
called $\UniquekSAT{3}$), one can derandomize the best known
randomized algorithm. We hope that with further effort one can
implement local search deterministically and achieve
Sch\"oning's running time.\\

Third, the algorithm of Paturi, Pudl\'{a}k, Saks, and Zane~\cite{ppsz}
runs faster on formulas with exactly one satisfying assignment than on
formulas with multiple ones (in fact, this is not precisely true: the
{\em proved} running time is better; no one knows about the actual
running time). Is $\UniquekSAT{3}$ {\em inherently} easier than
general $3$-SAT? Although Calabro, Impagliazzo, Kabanets, and
Paturi~\cite{CIKP-Unique} have shown that for large $k$,
$\UniquekSAT{k}$ can be at most marginally easier than general
$k$-SAT, this does not say anything about $3$-SAT. 
}
\subsection*{Notation}

A CNF formula $F$ is a conjunction (AND) of clauses. A clause is a
disjunction (OR) of literals, and a literal is either a boolean
variable $x$ or its negation $\bar{x}$. A $k$-clause is a clause with
exactly $k$ distinct literals. For example, $(x \vee \bar{y} \vee
\bar{z})$ is a typical $3$-clause. A $k$-CNF formula is a CNF formula
consisting of $k$-clauses. An {\em assignment} is a mapping of
variables to truth values. We use the numbers $0$ and $1$ to denote
$\texttt{false}$ and $\texttt{true}$. If $V$ is a set of $n$
variables, an assignment to $V$ can be seen as a bit string of length
$n$, i.e., an element of $\{0,1\}^n$. For two such assignments
$\alpha, \beta$, we denote by $d_H(\alpha,\beta)$ their {\em Hamming
  distance}, i.e., the number of variables on which they disagree.

\subsection{Randomized Local Search}
We will briefly describe Sch\"oning's local search
algorithm~\cite{Schoening99}. Let $F$ be a $3$-CNF formula and let
$\alpha$ be some truth assignment to the variables of $F$.  If
$\alpha$ does not satisfy $F$, the algorithm arbitrarily picks a
clause $C$ that is unsatisfied by $\alpha$. For simplicity, suppose $C
= (\bar{x} \vee \bar{y} \vee \bar{z})$. The algorithm uniformly at
random picks a variable from $C$ and flips the value $\alpha$ assigns
to it. This is repeated $3n$ times. If the algorithm did not find a
satisfying assignment within $3n$ steps, it gives up.
Sch\"oning~\cite{Schoening99} proved the following lemma.
\begin{lemma}[Sch\"oning~\cite{Schoening99}]
  Suppoe $F$ is a satisfiable $3$-CNF formula and $\alpha^*$ is a
  satisfying assignment. Let $\alpha$ be a truth assignment and set $r
  := d_H(\alpha,\alpha^*)$. Then the above algorithm finds a
  satisfying assignment with probability at least $(1/2)^r$.
\end{lemma}
Sch\"oning now picks the initial assignment $\alpha$ uniformly at
random from all truth assignments and then starts the local search we
described above. For a fixed satisfying assignment $\alpha^*$, we 
observe that 
$$
\Pr[d_H(\alpha, \alpha^*) = r] = \frac{{n \choose r}}{2^n} \ ,
$$
and therefore overall success probability is at least
$$
\sum_{r=0}^n \frac{{n \choose r}}{2^n} \left(\frac{1}{2}\right)^r = 
  \left(\frac{3}{4}\right)^n \ . 
$$
By repetition, this yields a Monte Carlo algorithm of running time
$O^*( (4/3)^n)$.  We see that Sch\"oning's algorithm uses
randomness in two ways: First to choose the initial assignment, and
second to steer the local search. It turns out that one can derandomize
the first random choice at almost no cost (the running time grows by a
polynomial factor in $n$). Derandomizing local search itself, however,
is much more difficult, and all currently known versions yield a
running time that is exponentially worse than Sch\"oning's randomized
running time.

\subsection{Deterministic Local Search}

We sketch the result Dantsin et al.~\cite{dantsin}, who were the first
to give a deterministic algorithm based on local search, which runs in
time $O^*(1.5^n)$. For that, consider the parametrized problem
$\BkS{3}$.
\begin{quotation}
  \textbf{$\BkS{3}$:} Given a $3$-CNF formula $F$ over $n$ variables,
  a truth assignment $\alpha$ to these variables, and a natural number
  $r$. Is there an assignment $\alpha^*$ satisfying $F$ such that
  $d_H(\alpha,\alpha^*) \leq r$?
\end{quotation}
We call this problem $\BkS{3}$ because it asks whether the Hamming
ball of radius $r$ around $\alpha$, denoted by $B_r(\alpha)$, contains
a satisfying assignment. The merits of~\cite{dantsin} are twofold.
First, they give a simple recursive deterministic algorithm solving
$\BkS{3}$ in time $O^*(3^r)$; If $\alpha$ does not satisfy $F$, pick
an unsatisfied clause $C$. There are $|C|\leq 3$ ways to locally
change $\alpha$ as to satisfy $C$. Recurse on each. One can regard
this recursive algorithm as a derandomization of Sch\"oning's local
search algorithm. It comes at a cost, however: Its running time is
$O^*(3^r)$, whereas Sch\"oning's local search has a success
probability of $(1/2)^r$. Therefore, a ``perfect'' derandomization
should have running time $O^*(2^r)$. Second, they show that an
algorithm $\mathcal{A}$ solving $\BkS{3}$ in time $O^*(a^r)$ yields an
algorithm $\mathcal{B}$ solving $3$-SAT in time
\begin{eqnarray}
O^*\left(\left(\frac{2a}{a+1}\right)^n\right) \ ,
\label{ball-to-sat}
\end{eqnarray}
and furthermore, $\mathcal{B}$ is deterministic if $\mathcal{A}$ is.
This works by covering the set of satisfying assignments with Hamming
balls of radius $r$ and solving $\BkS{3}$ for each ball.  Formally,
one constructs a {\em covering code} $\mathcal{C}$ of radius $r$,
which is a set $\mathcal{C} \subseteq \{0,1\}^n$ such that
$$
\bigcup_{\alpha \in \mathcal{C}} B_r(\alpha) = \{0,1\}^n \ 
$$
and then solves $\BkS{3}$ for each $\alpha \in \mathcal{C}$.\\

\ignore{ Sch\"oning's local search algorithm can be seen as a
  randomized algorithm solving $\BkS{3}$ in time $O^*(2^r)$ in the
  following sense: If $\alpha$ is a truth assignment, and there is
  some satisfying assignment $\alpha^*$ of $F$ with
  $d_H(\alpha,\alpha^*) \leq r$, then Sch\"oning's random walk
  algorithm, started on assignment $\alpha$, returns a satisfying
  assignment with probability at least $(1/2)^r$.  By repetition, this
  can be turned into a Monte Carlo algorithm with expected running
  time $O^*(2^r)$, which is much faster than the $O^*(3^r)$ by Dantsin
  et al.  The catch is that the satisfying assignment $\beta$ returned
  by Sch\"oning's local search may have $d_H(\alpha,\beta) > r$. In
  this sense, Sch\"oning's local search algorithm does not solve the
  decision problem $\BkS{3}$, but rather
  a promise version of it.\\
}
Can one solve $\BkS{3}$ {\em deterministically} in time $O^*(2^r)$?
Nobody has achieved that yet, although a lot of progress has been
made.  By devising clever branching rules (and proving non-trivial
lemmas), one can reduce the running time to $O^*(a^r)$ for $a < 3$.
Dantsin et al.  already reduce it to $O^*(2.848^r)$, Brueggeman and
Kern~\cite{BK04} to $O^*(2.792^r)$, and Scheder~\cite{Scheder08} to
$O^*(2.733^r)$. The approach which we present here is different.
Instead of designing new branching rules, we transform worst-case
instances of $\BkS{3}$ into $(3,3)$-CSP formulas over $r$ variables,
which one can solve more efficiently.  Of course, in reality several
subtleties arise, and the algorithm becomes somewhat technical.
Still, we achieve a substantially better running time:

\begin{theorem}
  $\BkS{3}$ can be solved deterministically in time $O^*(a^r)$, where
  $a \approx 2.562$ is the largest root of $x^2 - x - 4$.
  Together with (\ref{ball-to-sat}), this gives 
  a deterministic algorithm solving $3$-SAT in time
  $O^*\left(1.439^n\right)$.
\label{main-theorem}
\end{theorem}

\section{An Improved Algorithm for $\BkS{3}$}

In this section we will describe an algorithm for $\BkS{3}$. At
first, our algorithm for $\BkS{3}$ is not much different from the one
in Dantsin et al. For simplicity assume $\alpha = (1,\dots,1)$, and we
want to decide whether there is a satisfying assignment $\alpha^*$
that sets at most $r$ variables to $0$. We will describe a recursive
algorithm. By $L(r)$ we denote the number of leaves of its recursion
tree.
\subsection*{Intersecting Unsatisfied Clauses}
Suppose $F$ contains two negative $3$-clauses that intersect in one
literal, for example $(\bar{x} \vee \bar{y} \vee \bar{z})$ and
$(\bar{x} \vee \bar{u} \vee \bar{v})$. The algorithm has one ``cheap'' choice,
namely setting $x$ to $0$, and four ``expensive'' choices, namely
setting one of $y,z$ and one of $u,v$ to $0$. Recursing on all
five possibilities yields the recurrence
$$
L(r) \leq L(r-1) + 4L(r-2) \ .
$$
Standard methods show that this
recurrence has a solution in $O(a^r)$, with $a\approx 2.562$ being the
largest root of $x^2-x-4$.  If $F$ contains two negative $3$-clauses
intersecting in two literals, we obtain an even better recurrence: Let
$(\bar{x} \vee \bar{y} \vee \bar{z})$ and $(\bar{x}\vee \bar{y}\vee
\bar{u})$ be those two clauses. The algorithm has two cheap choices,
namely setting $x$ to $0$ or $y$ to $0$. Besides this, it has one
expensive choice, setting $z$ {\em and} $u$ to $0$.  This gives the
following recurrence:
$$
L(r) \leq 2L(r-1) + L(r-2) \ .
$$
This recurrence has a solution in $O\left( \left(\sqrt{2}+1\right)^r \right)
\leq O(2.415^r)$.
\subsection{Disjoint Unsatisfied $3$-Clauses}
Let $\Neg(F)$ denote the set of negative clauses in $F$, i.e., clauses
with only negative literals. Above we showed how to handle the case in
which $\Neg(F)$ contains intersecting clauses. From now on, unless
stated otherwise, we will assume that $\Neg(F)$ consists of pairwise
disjoint negative $3$-clauses. In all previous improvements to
deterministic local search, the case where $\Neg(F)$ consists of $r$
pairwise disjoint $3$-clauses constitutes the worst case. Somewhat
surprisingly, we can solve this case rather quickly, using a
deterministic algorithm for $(3,3)$-CSP by Scheder~\cite{Scheder-CSP}.
We call an assignment {\em exact} if it sets exactly $|\Neg(F)|$
variables to $0$, namely exactly one in each clause in $|\Neg(F)|$.
Here are two simple observations: (i) if $|\Neg(F)| > r$, then
$B_r(1,\dots,1)$ contains no satisfying assignment; (ii) if
$|\Neg(F)|=r$, then every satisfying assignment in $B_r(1,\dots,1)$ is
exact.

\begin{lemma}
  Suppose $\Neg(F)$ consists of $r$ pairwise disjoint $3$-clauses.
  Then $\BkS{3}$ can be solved in time $O^*(2.077^r)$
\end{lemma}
\begin{proof}
  If $\Neg(F)$ consists of $r$ pairwise disjoint $3$-clauses, then $F$
  has a satisfying assignment $\alpha^*$ in $B_r{(1,\dots,1)}$ if and
  only if it has an exact satisfying assignment. An exact assignment
  can satisfy each $C \in \Neg(F)$ in three different ways: Through
  its first, second, or third literal. We introduce a ternary variable
  $x_C$ to represent these three choices.  For example, if $C =
  (\bar{x} \vee \bar{y} \vee \bar{z})$, then every occurrence of the
  literal $y$ can be replaced by the literal $(x_C \ne 2)$, and a
  literal $\bar{y}$ can be replaced by $(x_C \ne 1 \wedge x_C \ne 3)$.
  If $u \in \vbl(F) \setminus \vbl(\Neg(F))$, we can replace the
  literal $u$ by \texttt{true} and $\bar{u}$ by \texttt{false}, as
  every exact assignment sets $u$ to $1$.  In this manner, we
  translate an instance of $\BkS{3}$ into a $(3,3)$-CSP problem over
  $r$ variables. By a result of~\cite{Scheder-CSP}, one can solve this
  in time $O^*(2.077^r)$.
\end{proof}
Let $m := |\Neg(F)|$. We have just seen that the case $m = r$ is
relatively easy.  If $m < r$ then we have a ``surplus budget'' of
$r-m$ that we can spend on satisfying multiple literals in some $C \in
\Neg(F)$, or setting variables $u$ to $0$ that do not occur in
$\Neg(F)$ at all.  This will make things more complicated and lead to
a worse running time:
\begin{lemma}
  Suppose $\Neg(F)$ consists of pairwise disjoint $3$-clauses.  Then
  $\BkS{3}$ can be solved in time $O^*(b^r)$, where
  $$
  b  = \frac{(5 + \sqrt{57})^2}{4(8 + \sqrt{57})} \approx 2.533 \ .
  $$ 
  \label{main-lemma}
\end{lemma}
In the rest of this section we will prove the lemma.  Write $V' :=
\vbl(F) \setminus \vbl(\Neg(F))$. Let $C=(\bar{x} \vee \bar{y} \vee
\bar{z}) \in \Neg(F)$. An assignment to the variables $x,y,z$ can be
represented by a string in $\{0,1\}^3$.  Seven of them satisfy $C$:
$011$, $101$, $110$, $001$, $010$, $100$, and $000$. We call them {\em
  colors}. The first three colors are {\em exact}, the latter four
{\em dirty}. Any assignment satisfying $\Neg(F)$ induces a
$7$-coloring of $\Neg(F)$. An assignment is exact if and only if it
sets every $u \in V'$ to $1$ and assigns every negative clause an
exact color (i.e., $011$, $101$, or $110$). We define a graph $G$ on
the seven colors, see Figure~\ref{figure-seven}.
\begin{figure}
  \begin{center}
    \epsfig{file=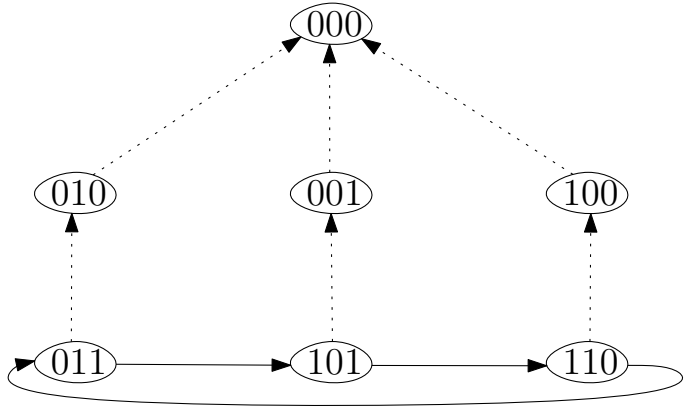,width=0.6\textwidth}
  \end{center}
  \caption{A graph on the seven colors with two kinds of edges.
  The seven colors represent the seven ways to satisfy
  a clause $(\bar{x} \vee \bar{y} \vee \bar{z})$.}
  \label{figure-seven}
\end{figure}
This graph has two types of edges: solid and dotted ones.
\begin{definition}
  For two colors $c,c'$, let $d(c,c')$ be the minimum number of solid
  edges on a directed path from $c$ to $c'$ (and $\infty$ if no such
  path exists), and $\cost(c,c')$ the minimum number of dashed edges
  ($\infty$ if no such path exists).
\end{definition}
For example $d(011, 100)=2$ and $\cost(011,100)=1$, but
$d(010,011)=\cost(010,011)=\infty$. Let $\alpha, \beta$ be two
assignments that satisfy $\Neg(F)$.  Recall that $\alpha$ and $\beta$
induce a $7$-coloring of $\Neg(F)$, thus for $C \in \Neg(F)$, we write
$\alpha(C), \beta(C)$ to denote this color.  To a variable $x \in V'$,
the assignment $\alpha$ does not assign one of the seven colors, but
simply a truth value, $0$ or $1$. To simplify notation, we define
$d(0,1)=d(1,0)=0$, $\cost(1,0)=1$, and $\cost(0,1)=\infty$.  Just
think of a graph on vertex set $\{0,1\}$ with a dotted
edge from $1$ to $0$, and no edge from $0$ to $1$ 
(Figure~\ref{figure-graph-two}).\\

\begin{figure}
  \begin{center}
    \epsfig{file=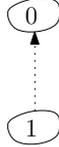,width=0.05\textwidth}
  \end{center}
  \caption{A graph on the two truth values $\{0,1\}$, representing
    the two ways to set a variable $x \in V'$.}
  \label{figure-graph-two}
\end{figure}
We define a ``horizontal distance'' from $\alpha$ to $\beta$
\begin{eqnarray}
d(\alpha,\beta) := \sum_{C \in \Neg(F)} d(\alpha(C),\beta(C)) \ .
\end{eqnarray}
and a ``vertical distance'' from $\alpha$ to $\beta$
\begin{eqnarray}
\cost(\alpha,\beta) := \sum_{C \in \Neg(F)} \cost(\alpha(C),\beta(C)) 
+ \sum_{x \in V'} \cost(\alpha(x),\beta(x)) \ .
\label{def-cost}
\end{eqnarray}
We say {\em horizontal} and {\em vertical} because these terms
correspond to the orientation of the two types of edges in
Figure~\ref{figure-seven}, as the solid edges are horizontal and the
dotted edges are vertical. 
\begin{proposition}
  Let $\alpha, \beta$ be two assignments satisfying $\Neg(F)$.
  \begin{enumerate}
    \item If $\alpha, \beta$ are exact,
    then $d(\alpha,\beta) \leq r$ and $\cost(\alpha,\beta)=0$.
    \label{prop-distance-cost-i}
    \item If $\alpha$ is exact, then $d(\alpha,\beta) \leq r$
      and $\cost(\alpha,\beta)$ is finite.
    \item If $B_r(1,\dots,1)$ contains an assignment $\alpha^*$
      satisfying $F$, then there is an exact assignment $\beta$ such
      that $d(\beta,\alpha^*) = 0$ and $\cost(\beta,\alpha^*) \leq
      r-m$. \label{prop-distance-cost-iii}
    \end{enumerate}
\label{prop-distance-cost}
\end{proposition}
We want to mimic the $(3,3)$-CSP of~\cite{Scheder-CSP}, although in
the case $|\Neg(F)| < r$, we cannot directly translate our instance of
$\BkS{3}$ into an instance of $(3,3)$-CSP. The idea, roughly speaking,
is to imitate the $(3,3)$-CSP-algorithm on the exact colors, while
using a traditional branching algorithm to search through dirty
colors. Let us be more precise:\\

We cover the set of exact assignments by a good covering code
$\mathcal{C}$. Formally, we want a (small) set $\mathcal{C}$
of exact assignments such that for every exact $\beta$, there is some
$\alpha \in \mathcal{C}$ such that $d(\alpha,\beta) \leq s$, where $s
\in \mathbb{N}$ is a suitably chosen integer. 
\begin{proposition}
  Let $\mathcal{C}$ as described. Then for every satisfying assignment
  $\alpha^* \in B_r(1,\dots,1)$, there exists an exact assignment
  $\alpha \in \mathcal{C}$ such that $d(\alpha,\alpha^*) \leq s$ and
  $\cost(\alpha,\alpha^*) \leq r-m$.
\end{proposition}
\begin{proof}
  By point (\ref{prop-distance-cost-iii}) of
  Proposition~\ref{prop-distance-cost}, there is an exact assignment
  $\beta$ such that $d(\beta,\alpha^*) = 0$ and $\cost(\beta,\alpha^*)
  \leq r-m$. By the properties of $\mathcal{C}$, there is an $\alpha
  \in \mathcal{C}$ such that $d(\alpha,\beta) \leq s$. By 
  point (\ref{prop-distance-cost-i}) of
  Proposition~\ref{prop-distance-cost}, $\cost(\alpha,\beta)=0$.
  Since $d$ and $\cost$ obey the triangle inequality (which is 
  easy to verify), it follows that
  $d(\alpha,\alpha^*) \leq s$ and $d(\alpha,\alpha^*) \leq r-m$.
\end{proof}
The main idea behind the algorithm of Dantsin et al. was
to focus not on $3$-SAT itself but on the parametrized problem
$\BkS{3}$. We will do the same here. We define 
the following decision problem: \\

\textbf{$\DBS$.} Given a $3$-CNF formula with $m \leq r$ pairwise
disjoint negative clauses, an assignment $\alpha$ satisfying
$\Neg(F)$, and two integers $s$ and $t$: Is there an assignment
$\alpha^*$ satisfying $F$
such that $d(\alpha, \alpha^*) \leq s$ and $\cost(\alpha,\alpha^*) \leq t$?\\

There are two special cases in which $\DBS$ can be solved rather
quickly. First, consider the case $t=0$. If $C$ is a clause
unsatisfied by $\alpha$, we have three possibilities to change
$\alpha$: At the first, second, or third literal of $C$. However, we
change $\alpha$ only along solid (horizontal) edges, since $t=0$
anyway rules out any assignment $\alpha^*$ differing from $\alpha$ by
a dotted (horizontal) edge. Since every vertex in $G$ is left by at
most one solid edge, the running time is $O^*(3^s)$. This
algorithm is in fact identical to the deterministic local search 
algorithm for $\BdkCSP{3}{3}$ in~\cite{Scheder-CSP}.\\

Second, consider the case $s=0$. If $\alpha$ does not satisfy clause
$C$, we again have three possibilities to change $\alpha$. This time,
however, solid edges are ruled out by $s=0$. Since every vertex
(color) is left by at most one dotted edge, this yields a running time
of $O^*(3^t)$.  In fact, the running time is better, namely
$O^*(2^t)$: Let $C$ be a clause not satisfied by $\alpha$.
Clearly $C \not \in \Neg(F)$, since $\alpha$ satisfies $\Neg(F)$.
Hence $C$ has at least one positive literal $x$. But look at
Figures~\ref{figure-seven} and~\ref{figure-graph-two}: Following a
dotted edge always means setting one additional variable to $0$, never
setting one to $1$. Therefore, the positive literal $x$ cannot become
satisfied by following a dotted edge, and there are actually at most
two choices to change $\alpha$, resulting in a running time of 
$O^*(2^t)$.\\

Let us summarize: The problem $\DBS$ with parameters $s$ and $t$ can
be solved in time $O^*(2.077^s)$ if $t=0$ and $O^*(2^t)$ if $s=0$. It
would be nice if these two border cases combined into a general running
time of $O^*(2.077^s2^t)$. Alas, this is not true. Or rather we do not
know how.

 \subsection*{An Algorithm for $\DBS$}
 
 We give a recursive algorithm $\dbs(F, \alpha, s,t)$ 
 for $\DBS$. We start with an assignment
 $\alpha$ satisfying $\Neg(F)$. As long as $\alpha$ does not 
 satisfy $F$, and $s,t \geq 0$, we modify $\alpha$ locally,
 in the hope of coming closer to a satisfying assignment, and
 continue recursively. \\
 
 There are two simple base cases. First, if $s<0$ or $t<0$, the
 algorithm returns \texttt{failure}. If $s,t\geq 0$ and $\alpha$
 satisfies $F$, it returns $\alpha$. Otherwise, $s,t \geq 0$, and
 there is some clause $C \in F \setminus \Neg(F)$ which $\alpha$ does
 not satisfy (recall that $\alpha$ satisfies $\Neg(F)$).  The clause
 $C$ has at most three literals, and at most two of them are negative.
 For each literal $\ell \in C$, the algorithm modifies $\alpha$ at one
 position in order to satisfy $\ell$. Note that each color has at most
 two outgoing edges, so to satisfy $\ell$, there are at most two
 direct ways to change the value of $\ell$ under $\alpha$. This means
 that each literal entails at most two recursive calls. The exact
 nature of these calls depends on the literal itself.\\
  
 We investigate which recursive calls are necessary when we try to
 change the value $\alpha$ assigns to $\ell \in C$. Note that $\alpha$
 currently does not satisfy $\ell$.  There are two cases: Either
 $\vbl(\ell) \in \vbl(\Neg(F))$ or $\vbl(\ell) \in V'$ (above, we
 defined $V'$ to be $\vbl(F) \setminus \vbl(\Neg(F))$). We start with
 the less interesting case.\\

 \textbf{Case 1.} $\vbl(\ell) \in V'$.\\

 \textbf{Case 1.1} $\ell$ is a positive literal, i.e., $\ell = x$ for
 some $x \in V'$. In this case, $\alpha(x)=0$, and for $x$, no edge
 leads back from $0$ to $1$. The algorithm gives
 up in this branch.\\
 
 \textbf{Case 1.2} $\ell$ is a negative literal, i.e., 
 $\ell = \bar{x}$ for some $x \in V'$. In this case,
 $\alpha(x)=1$, and there is only one way to change $\alpha$:
 Take the dotted edge, setting $x$ to $0$. The algorithm
 takes one recursive call:
 $$
 	\dbs(F,\alpha[x \rightarrow 0], s, t-1) \ .
 $$
 \textbf{Case 2.} $\vbl(\ell) \in \vbl(\Neg(F))$.  In this case, there
 is exactly one clause $D = (\bar{x} \vee \bar{y} \vee \bar{z}) \in
 \Neg(F)$ such that $\ell \in \{x,y,z,\bar{x},\bar{y},\bar{z}\}$.\\
 
 \textbf{Case 2.1} Suppose $\alpha(D)$ is dirty.
 There is only one outgoing edge, which is dotted, leading
 to some assignment $\alpha'$. Hence there is at most one
 recursive call, regardless of the literal $\ell$:
 $$
 	\dbs(F, \alpha', s,t-1) \ .
 $$
 In the remaining cases, we can assume that $\alpha(D)$ is pure.
 Without loss of generality, we assume that $\alpha(D) = 011$.  This
 means that $\alpha(x) = 0$ and $\alpha(y)=\alpha(z)=1$. Since
 $\alpha$ does not satisfy $\ell$, we conclude that $\ell \in \{x,
 \bar{y},\bar{z}\}$, which gives rise to three cases. As it will become
 clear soon, these three cases
 are the interesting cases in our analysis, whereas Cases 1.1, 1.2, and
 2.1 can be ignored---by our analysis, of course, not by the algorithm.\\
  
 \textbf{Case 2.2} $\ell = x$. Dotted edges only set variables to $0$, 
 so they are of no help here. The algorithm changes $\alpha(D)$ to $101$ by
 choosing the outgoing solid edge, and calls itself recursively:
 $$
 	\dbs(F, \alpha[D \rightarrow 101], s-1, t)   \ .
 $$
 
 \textbf{Case 2.3} $\ell = \bar{y}$. We have several possibilities. A
 satisfying assignment that satisfies $\bar{y}$ could set $D$ to
 $101$, $001$, $100$, or $000$. We do not want to cause four recursive
 calls. Note that if a satisfying assignment $\alpha^*$  sets
 $D$ to $101$, $001$, or $100$, then changing $\alpha(D)$ from $011$
 to $101$ decreases $d(\alpha,\alpha^*)$ by $1$ in any case. If
 $\alpha^*(D) = 000$, then changing $\alpha(D)$ to $000$ decreases
 $\cost(\alpha,\alpha^*)$ by $2$. Therefore the algorithm calls itself
 twice:
 \begin{eqnarray*}
 	& \dbs(F, \alpha[D \rightarrow 101],s-1,t)  &  \\ 
	&  \dbs(F, \alpha[D \rightarrow 000],s, t-2)  &
 \end{eqnarray*}
 
 \textbf{Case 2.4} $\ell = \bar{z}$. Observe that $\alpha^*(D)$ is
 either $110$, $010$, $100$, or $000$. If $\alpha^*(D)$ is $110$ or
 $100$, then changing $\alpha(D)$ from $011$ to $110$ decreases
 $d(\alpha,\alpha^*)$ by $2$. If $\alpha^*(C)$ is $010$ or $000$,
 changing $\alpha(D)$ from $011$ to $010$ decreases
 $\cost(\alpha,\alpha^*)$ by 1. The algorithm thus calls itself twice:
 \begin{eqnarray*}
 	& \dbs(F,\alpha[D \rightarrow 110], s-2, t) & \\
	& \dbs(F,\alpha[D \rightarrow 010], s, t-1)  &
 \end{eqnarray*}
 
 For the analysis of the running time, we can ignore Cases 1.1, 1.2, 
 and 2.1, since by Cases 2.2--2.4.  Let $L(s,t)$ be the worst-case
 number of leaves in a recursion tree of $\dbs(F,\alpha,s,t)$.
 \begin{proposition}
   If $s < 0$ or $t < 0$, then $L(s,t) = 1$.
   Otherwise, 
   \begin{eqnarray*}
     L(s,t) \leq L(s-1,t) + 2 \max\left(
       \begin{array}{l}
         L(s-1,t) + L(s,t-2) , \\
         L(s-2,t) + L(s,t-1)
       \end{array}
       \right) \ .       
   \end{eqnarray*}
   \label{proposition-L-s-t}
 \end{proposition}
 \begin{proof}
   The proof works by induction. If $s < 0$ or $t < 0$, then
   clearly there is no assignment $\alpha^*$ with
   $d(\alpha,\alpha^*) \leq s$ and $\cost(\alpha,\alpha^*) \leq t$,
   and the algorithm simply returns failure.\\

   If $s,t \geq 0$, let $C \in F \setminus \Neg(F)$ be a clause that
   $\alpha$ does not satisfy. In the worst case, $|C| = 3$ and
   consists of one positive and two negative literals. There are three
   cases: (i) the negative literals can be both of type $\bar{y}$
   (Case 2.3), (ii) both of type $\bar{z}$ (Case 2.4), or (iii) one of
   type $\bar{y}$ and one of $\bar{z}$.  If (i) holds, then we 
   can bound $L(s,t)$ by
   \begin{eqnarray}
     3 L(s-1,t) + 2L(s,t-2) \ .
   \label{ineq-xyy}
   \end{eqnarray}
   If (ii) holds, we can bound $L(s,t)$ by 
  \begin{eqnarray}
    L(s-1,t) + 2 L(s-2,t) + 2L(s,t-1) \ .
    \label{ineq-xzz} 
  \end{eqnarray} 
  Finally, if (iii) holds, we bound $L(s,t)$ by 
  $$
  L(s-1,t) + (L(s-1,t)+L(s,t-2)) + (L(s-2,t)+L(s,t-1)) \ ,
  $$
  but one easily verifies that this is bounded from above
  by either (\ref{ineq-xyy}) or (\ref{ineq-xzz}).
\end{proof}

\begin{lemma}
  Let $a,b \geq 1$ be such that 
  \begin{eqnarray}
    ab^2 & \geq & b^2 + 2a \label{ineq-ab-xyy} \\
    & \textnormal{and} & \nonumber \\
    a^2b & \geq & ab + 2a^2 + 2b \label{ineq-ab-xzz} \ .
  \end{eqnarray} 
  Then $L(s,t) \in O(a^sb^t)$.
\end{lemma}
\begin{proof}
  We use induction to show that $L(s,t) \leq Ca^sb^t$ for some
  sufficiently large constant $C$. For the base case, i.e., $s<0$ or
  $t<0$, choose a constant $C$ large enough that $1=L(s,t) \leq
  Ca^sb^t$. If $s,t \geq 0$, then conditions (\ref{ineq-ab-xyy}) and
  (\ref{ineq-ab-xzz}) guarantee that the induction goes through when
  one uses the bound from Proposition~\ref{proposition-L-s-t}.
\end{proof}

\subsection*{Combining $\dbs$ with Covering Codes}

Our overall algorithm works follows: If $\Neg(F)$ consists of at most
$r$ pairwise disjoint negative $3$-clauses, it constructs a covering
code $\mathcal{C}$ of radius $s$ for the set of exact assignments. In
other words, $\mathcal{C}$ is such that for every exact assignment
$\beta$, there is some exact assignment $\alpha \in C$ with
$d(\alpha,\beta) \leq s$. Here, $s$ is some natural number to be
determined later. It then calls
$$
\dbs (F, \alpha, s, m) 
$$
for each $\alpha \in \mathcal{C}$, where $m := r-|\Neg(F)|$ is our
``surplus budget''. If no run of $\dbs$ finds a satisfying assignment,
it concludes that $B_r(1,\dots,1)$ contains no satisfying assignment,
and returns failure.  The overall running time of this business is
\begin{eqnarray}
|\mathcal{C}| a^s b^{r-m} \poly(n) \ .
\label{ineq-running-time-code}
\end{eqnarray}
The following lemma is from Scheder~\cite{Scheder-CSP}, adapted
to our current terminology.
\begin{lemma}[\cite{Scheder-CSP}]
  For every $x > 0$, there is some $s \in \{0,1,\dots,2m\}$ and a
  covering code $\mathcal{C}$ of radius $s$ for the set of exact
  assignments of size
  $$
  |\mathcal{C}| \leq \frac{3^m x^s}{(1+x+x^2)^m} \poly(m) \ .
  $$
  Furthermore, one can deterministically construct 
  $\mathcal{C}$ in time $O(|\mathcal{C}|)$.
\end{lemma}
Combining (\ref{ineq-running-time-code}) with the lemma and
setting $x := 1/a$ in the lemma, we see that the running 
time of the algorithm is at most 
\begin{eqnarray}
\frac{3^m x^s a^s}{(1+x+x^2)^m} b^{r-m} \poly(n)
= \left(\frac{3a^2}{a^2+a+1}\right)^m b^{r-m} \poly(n) \ .
\label{running-time-m}
\end{eqnarray}
We still can choose $a$ and $b$, as long as they satisfy
(\ref{ineq-ab-xyy}) and (\ref{ineq-ab-xzz}). We try the following: We
guess (with hindsight, we know) that (\ref{ineq-xzz}) dominates the
running time, thus we try to satisfy (\ref{ineq-ab-xzz}) with
equality. Furthermore, we want to get rid of the parameter $m$ in
(\ref{running-time-m}), which depends on the formula $F$ and over
which we do not have control. In other words, we want to choose $a$
and $b$ such that
\begin{eqnarray*}
  a^2b & = & ab + a^2 + 2b \\
  b & = & \frac{3a^2}{a^2+a+1}\ .
\end{eqnarray*}
One checks that $a = (5+\sqrt{57})/2$ and $b = (5+\sqrt{57})^2 / (4
(8+\sqrt{57}))$ will do, and also satisfy (\ref{ineq-ab-xyy}). With
these numbers, (\ref{running-time-m}) boils down to $O^*(b^r) \leq
O^*(2.533^r)$. This finishes the proof of Lemma~\ref{main-lemma}.\\

We observe that our algorithm solves the case where $\Neg(F)$ consists
of pairwise disjoint negative $3$-clauses more efficiently, in
$O^*(2.533^r)$, than the case where $F$ contains overlapping
negative clauses, in which there is a branching leading to a running
time of $O^*(2.562^r)$. This is qualitatively different from
all previous approaches to improving local search (Dantsin et
al.~\cite{dantsin}, Brueggeman and Kern~\cite{BK04} and
Scheder~\cite{Scheder08}): In those approaches, pairwise disjoint
negative clauses constitute the worst case, and the case where $F$
contains intersecting clauses is always the easy case handled at the
very beginning. Now the picture has changed: A further improvement
will have to work on that case, too.\\

We are ready to prove Theorem~\ref{main-theorem}.
\begin{proof}[Proof of Theorem~\ref{main-theorem}]
  The algorithm outlined above solves $\BkS{3}$ in time $O^*(a^r)$, 
  where $a \approx 2.562$, with the worst case being that
  $F$ contains negative clauses that intersect in one literal.
  Combining this with a standard construction of covering codes,
  we conclude that $3$-SAT can be solved in time
  $$
  \left(\frac{2a}{a+1}\right)^n \poly(n) \leq O^*\left(1.439^n\right) \ ,
  $$  
  which finishes the proof.
\end{proof}

\bibliographystyle{abbrv}
\bibliography{refs}

\end{document}